\documentclass{baustms}

\citesort

\theoremstyle{cupthm}
\newtheorem{thm}{Theorem}[section]

\newtheorem{cor}[thm]{Corollary}
\newtheorem{lem}[thm]{Lemma}
\theoremstyle{cupdefn}

\theoremstyle{cuprem}
\newtheorem{rem}[thm]{Remark}
\numberwithin{equation}{section}
\newtheorem{conj}[thm]{Conjecture}

\newcommand{\FF}{{\mathbb{F}}}

\newcommand{\seq}{\subseteq}
\newcommand{\ZZ}{\mathbb{Z}}

\newcommand{\cv}{\mathbf{c}}

\newcommand{\xv}{\mathbf{x}}
\newcommand{\yv}{\mathbf{y}}

\def\Item$#1${\item $\displaystyle#1$
   \hfill\refstepcounter{equation}(\theequation)}

\begin{document}
\runningtitle{New optimal linear codes over $\ZZ_4$}
\title{New optimal linear codes over $\ZZ_4$}
\author[1]{HOPEIN CHRISTOFEN TANG}
\address[1]{Combinatorial Mathematics Research Group,
	Facul-ty of Mathematics and Natural Sciences,
	Institut Teknologi Bandung,
	Jl. Ganesha 10, Bandung, 40132,
	INDONESIA\email{hopeinct@students.itb.ac.id}}
\cauthor
\author[2]{DJOKO SUPRIJANTO}
\address[2]{Combinatorial Mathematics Research Group,
	Faculty of Mathematics and Natural Sciences,
	Institut Teknologi Bandung,
	Jl. Ganesha 10, Bandung, 40132,
	INDONESIA\email{djoko.suprijanto@itb.ac.id}}

\authorheadline{H. C. Tang and D. Suprijanto}


\support{This research is supported by Institut Teknologi Bandung (ITB) and the Ministry of Education, Culture, Research and Technology
(\emph{Kementerian Pendidikan, Kebudayaan, Riset dan Teknologi (Kemdikbud-ristek)}),
Republic of Indonesia.}

\begin{abstract}
In this work, we introduce new methods to construct new linear codes over $\ZZ_4$ from the known ones. Many of these new codes are optimal. We obtained all optimal codes for $k_1=2,~k_2=0$ and many optimal codes for $k_1=3,~k_2=0.$  
\end{abstract}


\maketitle

\section{Introduction}



In early development, algebraic coding theory considered finite fields as alphabet codes. Codes over finite rings were introduced in the first half of the 70s decade by Blake \cite{Blake1, Blake2}.
He \cite{Blake1} showed how to construct codes over $\ZZ_m$ from cyclic codes
over $\FF_p,$ where $p$ is a prime factor of $m.$ He \cite{Blake2} then further observed
the structure of codes over $\ZZ_{p^r}.$ Spiegel \cite{Spiegel1, Spiegel2} generalized Blake's results
to codes over $\ZZ_m,$ where $m$ is an arbitrary positive integer.
The study of codes over finite rings attracted great interest in algebraic coding theory
through the work of Hammons, Kumar, Calderbank, Sloane, and Sol\'{e} \cite{Hammons1994},
where they show how several well-known families of nonlinear binary codes were
intimately related to linear codes over $\ZZ_4.$ Since Hammons et al. \cite{Hammons1994},
many people have been considering numerous aspects of codes over various finite rings.

Pless and Qian \cite{Pless2} considered cyclic, and particularly quadratic residue codes over $\ZZ_4.$ Cyclic, negacyclic, and also the mass formula for the codes over $\ZZ_4$ have been explored, e.g., by Abualrub and Oehmke \cite{Abualrub2003}, Blackford \cite{Blackford2003a, Blackford2003b}, 
to name a few.

Other aspects have also been investigated by many authors. For example, in 1997, Bonnecaze and Duursma \cite{Bonnecaze1} considered the translate of linear codes over $\ZZ_4.$ Moreover, they \cite{Bonnecaze1} also introduced a method to compute the complete weight enumerators of the translate codes. The decoding algorithm for linear codes over $\ZZ_4$ has been observed by
Byrne, Greferath, Pernas, and Zumbr\"{a}gel \cite{Byrne2013}. 

Regarding the construction of optimal codes over $\ZZ_4,$ so far only self-dual codes are available in the literatures. See, for instance,
\cite{Huffman2005} and \cite{Rains1999}.
On the other hand, direct constructions of linear (but not self-dual) codes over $\ZZ_4$ are still in the initial stage. There is not much literature on it, and among them are the works of Gulliver and Wong \cite{Gulliver2007}, Dougherty, Gulliver, Park, and Wong \cite{Dougherty2007}, and also Aydin and Asamov \cite{Aydin2009, Aydin2}. So far, the new linear codes over $\ZZ_4$ obtained undirectly, as a $\ZZ_4$-image of certain Gray maps from finite rings to $\ZZ_4,$ where in many cases, the resulting codes have not-so-good minimum distances (see, e.g., \cite{Bustomi21} for recent development in this direction).

The purpose of this paper is two folds. First, we reprove the Plotkin-type bound for the Lee distance of linear codes by an elementary approach. Second, we introduce new construction methods of linear codes over $\ZZ_4.$ By using these methods, we obtained many new free linear codes over $\ZZ_4$ with the highest known minimum Lee distance as well as optimal linear codes over $\ZZ_4.$ Our results contribute significantly to the Aydin and Asamov's database of $\ZZ_4$ codes \cite{Aydin2009, Aydin2}.

The organization of the paper is as follows. In Section 2, we describe several basic facts regarding linear codes over $\ZZ_4.$
Section 3 contains the proof of Plotkin-type bound for the Lee distance.
Several construction methods for linear codes over $\ZZ_4$ and their applications to obtain new optimal codes over $\ZZ_4$ as well as the linear codes over $\ZZ_4$ with the highest known minimum Lee distance are given in Section 4. This paper is ended by concluding remarks. We follow \cite{Huffman} for undefined terms in coding theory.

\section{Preliminaries}

A \emph{code} of length $n$ over the ring $\ZZ_4$ is a nonempty subset of $\ZZ_4^n.$ If the code is also a submodule of $\ZZ_4^n,$
then we say that the code is \emph{linear.} The linear code is called \emph{free} if it is a free submodule of $\ZZ_4^n.$ \textcolor{red}{We only consider nonzero linear codes.}


A matrix $G\in \ZZ_4^{k \times n}$ is called a \emph{generator matrix} of a linear code $C$ of length $n$ over $\ZZ_4$ if the rows of $G$ generate $C$ and no proper subset of the rows of $G$ generates $C$.

Two codes are said to be \emph{equivalent} if one can be obtained from the other by permuting the coordinates and (if necessary) changing the signs of certain coordinates. Codes differing by only a permutation of coordinates are called \emph{permutation-equivalent.} It is well-known (see \cite{Hammons1994}) that any linear code over $\ZZ_4$ is permutation-equivalent to the linear code $C$ with generator matrix $G$ of the form
\begin{equation}\label{G-standard}
	G=
	\begin{pmatrix}
		I_{k_1} & A & B_1+2 B_2\\
		0 & 2 I_{k_2} & 2 D
	\end{pmatrix},
\end{equation}
where $A,B_1,B_2,$ and $D$ are $(0,1)$-matrices. Moreover, the code $C$ is a free linear code if and only if $k_2=0.$ The generator matrix of a linear code $C$ over $\ZZ_4$ is called in a \emph{standard form} if it has the form as given in the Equation (\ref{G-standard}).

The \emph{Lee weight} of $ x \in \ZZ_4,$ denoted by $w_L(x),$ is defined by $w_L(0)=0,$ $w_L(1)=1,$ $w_L(2)=2,$ and $w_L(3)=1.$ The Lee weight of a vector $\xv=(x_1,x_2,\ldots,x_n) \in \ZZ_4^n$ is defined as
$
w_L(\xv)=\sum_{i=1}^n w_L(x_i).
$
For $\xv,\yv \in \ZZ_4^n,$ the Lee distance between $\xv$ and $\yv,$ denoted by $d(\xv,\yv),$ is defined by
$d(\xv,\yv)=w_L(\xv-\yv).$ The \emph{minimum Lee distance} of a linear code $C \seq \ZZ_4^n$ is defined by \[d_L=d_L(C):=\mathrm{min}\{d_L(\xv,\yv):~\xv,\yv \in C, \xv \neq \yv\}.\] It is clear that for a linear code $C,$ the minimum Lee distance is exactly the same as the minimum Lee weight, namely
$d_L(C)=\mathrm{min}\{ w_L(\xv):~\xv \in C, \xv \neq \mathbf{0}\}.$
We write the parameters of a linear code $C$
over $\ZZ_4$ as $[n, 4^{k_1}2^{k_2}, d_L],$ where $n$ is a length of $C,$ $|C| = 4^{k_1}2^{k_2},$ and $d_L = d_L(C).$ Following
Hammons et al. \cite{Hammons1994} (cf. \cite{Wan}), we say that $C$ is of type $4^{k_1}2^{k_2}.$

Determining or estimating the minimum distance of a given linear code is one of the fundamental problems in coding theory. It is because the minimum distance of the code is proportional to its capability in detecting and correcting errors that appear during the transmission of information.

In 2001, Dougherty and Shiromoto \cite{Dougherty2001} proved an upper bound for the minimum Lee distance of codes over $\ZZ_4$, as follows.

\begin{thm}[Singleton-type Lee distance bound (Theorem 3.1 in \cite{Dougherty2001})]
	If $C$ is a linear code of length $n$ over $\ZZ_4$ with parameters $[n,4^{k_1}2^{k_2},d_L]$, then
	\begin{equation}\label{LDB1}
			d_L\leq 2n-2k_1-k_2+1.
	\end{equation}

\end{thm}

\begin{rem}
	\emph{Dougherty and Shiromoto \cite{Dougherty2001} called the codes meeting the bound (\ref{LDB1}) as} maximum Lee distance separable (MLDS) \emph{codes.}
\end{rem}

Let $G \in \ZZ_4^{(k_1+k_2)\times n}$ be a generator matrix of a linear code $C$ over $\ZZ_4.$ Following Klove \cite{Klove}, for any $\cv\in\mathbb{Z}_4^{k_1+k_2},$ we define the \emph{multiplicity} of $\cv,$ denoted by $\mu(\cv),$ as the number of occurrences of $\cv$ as a column vector in $G.$ Observe that for any $\xv\in\mathbb{Z}_4^k$ with $k:=k_1+k_2$, we have
\[
w_L\left(\xv G\right)=\sum_{\cv\in\mathbb{Z}_4^k} \mu(\cv) w_L(\xv\cdot\cv).
\]
It is also clear that
\[
\sum_{\cv\in\mathbb{Z}_4^k} \mu(\cv)=n.
\]

\section{Plotkin-type Lee distance bound}

For $\xv=(x_1,x_2,\ldots,x_n),~\yv=(y_1,y_2,\ldots,y_n)\in\ZZ_4^n,$ we define the \emph{Euclidean inner product} of $\xv$ and $\yv$ as
$\xv\cdot\yv=\sum_{i=1}^n x_i y_i\in \ZZ_4.$
For $\xv\in\mathbb{Z}_4^n,$ let $\xv'=(x_1,\dots,x_{n-1}) \in \ZZ_4^{n-1},$ namely the vector obtained from $\xv$ by dropping its last entry. For the rest of this paper, let \textbf{0} denote the zero vector.

The lemma below shows that for any given nonzero codeword $\cv\in\ZZ_4^k$, the sum $\sum_{\xv\in\ZZ_4^k}w_L(\xv\cdot \cv)$ is constant, and depends only on $k.$ This result is very important to derive subsequent properties.

\begin{lem}\label{lemma1}
	Let $k$ be a positive integer. If $\cv=(c_1,c_2,\dots,c_k)\in\ZZ_4^k$ is a nonzero vector, then
	\[
	\sum_{\xv\in\ZZ_4^k}w_L(\xv\cdot \cv)=4^k.
	\]
\end{lem}

\begin{proof}
It is clear for the case $k=1$. Since $\cv$ is a nonzero vector, there exists $i$ such that $c_i\neq 0$. Without loss of generality, let $i=k$.
	If $\cv' = \mathbf{0},$ then we have
	\begin{align*}
		\sum_{\xv\in\mathbb{Z}_4^{k}}w_L(\xv\cdot \cv)&=\sum_{\xv\in\mathbb{Z}_4^{k}}w_L(x_{k}c_{k})=4^k.
	\end{align*}

If $\cv' \neq \mathbf{0},$ define $
H_i:=\{\yv\in\mathbb{Z}_4^{k-1}:~\yv\cdot\cv'=i\} \text{ for }i\in \ZZ_4$. Since $c_{k}\neq 0$, we have
	\begin{align*}
		\sum_{\xv\in\mathbb{Z}_4^{k}}w_L(\xv\cdot \cv)&=\sum_{i=1}^4\sum_{\xv'\in H_i}\sum_{x_k=0}^3 w_L(\xv'\cdot \cv'+x_{k}c_{k})\\
		&=4|H_0|+4|H_1|+4|H_2|+4|H_3|=4^k.
	\end{align*}
\end{proof}

By applying Lemma \ref{lemma1}, we have the following lemma.
\begin{lem}\label{lem}
	Let $n$ and $k$ be positive integers. If a matrix $G\in\mathbb{Z}_4^{k\times n}$ has no zero column, then
	\[
	\sum_{\xv\in\mathbb{Z}_4^k}w_L\left(\xv G\right)=4^k n.
	\]
\end{lem}


\begin{lem}[Constant summation of Lee weight]\label{Constant summation}
	Let $C$ be a linear code over $\ZZ_4$ with parameters $[n,4^{k_1}2^{k_2},d_L].$ If $\mu(\mathbf{0})=0$, then
	\[
	\sum_{\cv\in C} w_L(\cv)=|C|n.
	\]
\end{lem}

\begin{proof}
	Let $G$ be a generator matrix of $C$ in a standard form. Since $G$ does not have any zero column, then, by Lemma \ref{lem}, we have $\sum_{\xv\in\mathbb{Z}_4^{k}}w_L\left(\xv G\right)=4^{k} n$ for $k:=k_1+k_2$. If $k_2=0,$ then
	\[
	\sum_{\cv\in C} w_L(\cv)=\sum_{\xv\in\mathbb{Z}_4^{k_1}}w_L\left(\xv G\right)=4^{k_1} n=|C|n.
	\]
	If $k_2>0,$ let $\mathbf{g^{(1)}},\dots,\mathbf{g^{(k_2)}}$ be all distinct $k_2$ row vectors of $G$ which have entries $0$ or $2$ only. We have
	\[
	\xv G=\xv G+\sum_{i=1}^{k_2}x_i \mathbf{g^{(i)}}=\Big(\xv+[0,\dots,0,x_1,\dots,x_{k_2}]\Big)G
	\]
	for $x_1,\dots,x_{k_2}\in\{0,2\}.$ There are $2^{k_2}$ possible different values for $x_1,\dots,x_{k_2}.$ Hence, in the summation $\sum_{\xv\in\mathbb{Z}_4^{k}}w_L\left(\xv G\right),$
	every codeword in $C$ appears exactly $2^{k_2}$ times. Thus
	\[
	4^{k} n=\sum_{\xv\in\mathbb{Z}_4^{k}}w_L\left(\xv G\right)=2^{k_2}\sum_{\cv\in C} w_L(\cv).
	\]
	Therefore,
	\[
	\sum_{\cv\in C} w_L(\cv)=\frac{4^{k} n}{2^{k_2}}=4^{k_1}2^{k_2}n=|C|n.
	\]
\end{proof}

Next, we prove the Plotkin-type bound for the minimum Lee distance.

\begin{thm}[Plotkin-type Lee distance bound]\label{Plotkin}
	Let $C$ be a linear code over $\ZZ_4$ of parameters $[n,4^{k_1}2^{k_2},d_L].$ Then
	\begin{equation}\label{PLDB}
		d_L\leq \frac{|C|}{|C|-1}n.
	\end{equation}
\end{thm}

\begin{proof}
	Let $\mu(\mathbf{0}) = m$. By using Lemma \ref{Constant summation}, it is easy to see that
	$|C|(n-m)=\sum_{\cv\in C} w_L(\cv)\geq (|C|-1)d_L.$
	Therefore,
	\[
	d_L\leq \frac{|C|}{|C|-1}(n-m)\leq\frac{|C|}{|C|-1}n.
	\]
\end{proof}

The linear codes $C$ whose minimum Lee distance $d_L(C)$ is an integer nearest to the upper bound of the Plotkin-type Lee distance bound (\ref{PLDB}) as given in Theorem \ref{Plotkin} is called \emph{Plotkin-optimal.} In other words, a linear code $C$ is Plotkin-optimal if $\displaystyle d_L(C)=\left \lfloor \frac{|C|}{|C|-1}n \right \rfloor.$

\begin{rem}
	\emph{The bound similar to Theorem \ref{Plotkin} for (not necessarily linear) codes over $\ZZ_m$ has been proven by Wyner and Graham \cite{Wyner1968}. However, our proof above is very simple (compare with the one in \cite{Wyner1968}), although limited for only linear codes over $\ZZ_4.$}
\end{rem}

\section{Optimal codes}
It is well-known about the nonexistence of MLDS codes except the trivial ones. This fact was proved by Dougherty and Shiromoto \cite{Dougherty2001}.

\begin{thm}[\cite{Dougherty2001}]
	There is no MLDS codes over $\ZZ_4$ except the trivial ones, namely the linear codes having parameters $[n,4^02^1,2n],$ $[n,4^n2^0,1],$ or $[n,4^{n-1}2^1,2].$
\end{thm}


 Since there are no linear codes over $\ZZ_4$ having the minimum Lee distance attaining the Singleton-type Lee distance bound (\ref{LDB1}) except the trivial ones, then we consider the optimality of the codes with respect to the Plotkin-type Lee distance bound (\ref{PLDB}).  

\subsection{Optimal linear codes.}
In this section, we discuss several construction methods of optimal codes. For any given nonnegative integer $k_1,k_2$ with $k_1+k_2>0$, let $G^{(k_1,k_2)}$ denote the generator matrix of a linear code $C^{(k_1,k_2)}$ whose columns consist of all possible nonzero vectors in $\ZZ_4^{k_1}\times (2\ZZ_4)^{k_2}.$
We can define the matrix $G^{(k_1,k_2)}$ recursively as follows: \[G^{(0,1)} :=
\begin{bmatrix}
	2
\end{bmatrix}, G^{(1,0)} :=
\begin{bmatrix}
	1 &2 &3
\end{bmatrix},\]
\[
G^{(k_1+1,0)}:=\arraycolsep=3.0pt
\left[\begin{array}{@{}c|c|c|c| c c c@{}}
	G^{(k_1,0)} & G^{(k_1,0)} & G^{(k_1,0)} & G^{(k_1,0)} &\textbf{0} & \textbf{0} & \textbf{0} \\\hline
	0\;0\dots 0 &1\;1\dots 1 &2\;2\dots 2 &3\;3\dots 3 &1 &2 &3
\end{array}\;\right],
\]
\[G^{(k_1,k_2+1)}:=
\left[\begin{array}{@{}c|c|c @{}}
	G^{(k_1,k_2)} & G^{(k_1,k_2)} & \textbf{0}\\\hline
	0\;0\dots 0  &2\;2\dots 2 &2
\end{array}\;\right].
\]
\normalsize
The theorem below shows the existence of a linear code over $\ZZ_4$ of length $n$ satisfying the bound (\ref{PLDB}), for any given $k_1$ and $k_2.$  Moreover, the code has a constant Lee weight.

\begin{thm}[Existence of optimal codes]\label{konstruk}
	The linear code $C^{(k_1,k_2)}$ is of parameters $[n,4^{k_1}2^{k_2},d_L],$ with $n=4^{k_1}2^{k_2}-1,$ and $d_L=4^{k_1}2^{k_2}.$ Moreover, all nonzero codewords in $C^{(k_1,k_2)}$ have the Lee weight $d_L.$
\end{thm}

\begin{proof}
	First, consider the case $k_2=0.$ Observe that for any nonzero vector $\yv\in C^{(k_1,0)}$ there exists a unique $\xv\in\mathbb{Z}_4^{k_1}$ such that $\yv=\xv G^{(k_1,0)}$. By Lemma \ref{lemma1}, $w_L(\yv)$ is equal to
	\[
	w_L\left(\xv G^{(k_1,0)}\right)=\sum_{\cv\in\mathbb{Z}_4^{k_1}} \mu(\cv) w_L(\xv\cdot\cv)=\sum_{\cv\in\mathbb{Z}_4^{k_1}} w_L(\cv\cdot\xv)=4^{k_1}.
	\]
	For the case when $k_2>0,$ we do an induction on $k_2$. Observe that for any nonzero $\yv\in C^{(k_1,k_2)}$, there exists a unique $\xv\in\mathbb{Z}_4^{k_1}\times (2\mathbb{Z}_4)^{k_2}$ such that $\yv=\xv G^{(k_1,k_2)}$.
	
	Define $k:=k_1+k_2$. Notice that
	\begin{align*}
		w_L(\yv)
	&=\sum_{\cv\in\mathbb{Z}_4^{k}} \mu(\cv) w_L(\xv\cdot\cv)=\sum_{\cv'\in\mathbb{Z}_4^{k_1}\times (2\mathbb{Z}_4)^{k_2-1}} w_L(\cv'\cdot\xv')+w_L(\cv'\cdot\xv'+2x_{k}),
 	\end{align*}
and the last sum is equal to $4^{k_1}2^{k_2}$.
\end{proof}

The construction below is easy to verify.

\begin{lem}[Construction A]\label{konstruk1}
	If $C_1$ and $C_2$ are linear codes of type $4^{k_1}2^{k_2}$ and $n(C_1)+n(C_2)< 4^{k_1}2^{k_2}-1,$ having generator matrices (in standard form) $G_1$ and $G_2,$ respectively, then the linear code $C'$ generated by the matrix $G':=[G_1|G_2]$ is of length $n(C_1)+n(C_2)$ with $d_L(C')\geq d_L(C_1)+d_L(C_2)$.
\end{lem}

The following construction is a special case of Lemma \ref{konstruk1}. However, we mention it here for the reader's convenience.

\begin{cor}[Construction 1 for optimal codes]\label{opt1}
	If $C$ is a Plotkin-optimal linear code of parameters $[n,4^{k_1}2^{k_2},d_L]$ with generator matrix (in standard form) $G,$ then the linear code $C'$ generated by the matrix $G':=\left[G|G^{(k_1,k_2)}\right]$ is also optimal of length $n+4^{k_1}2^{k_2}-1.$
\end{cor}

\begin{proof}
	It is clear that $n(C')=n(C)+n\left(C^{(k_1,k_2)}\right)=n+4^{k_1}2^{k_2}-1$ and $|C|=|C'|$. Moreover, since the Lee weight of any nonzero codeword in $C^{(k_1,k_2)}$ is $4^{k_1}2^{k_2}$ (Theorem \ref{konstruk}), then
	$d_L(C')\geq d_L(C)+d_L\left(C^{(k_1,k_2)}\right)=\left\lfloor\frac{|C|}{|C|-1}n(C')\right\rfloor$. Therefore, $C'$ is optimal by Theorem \ref{Plotkin}.
\end{proof}

\begin{cor}[Construction 2 for optimal codes]\label{opt2}
	If $C_1$ and $C_2$ are Plotkin-optimal linear codes of type $4^{k_1}2^{k_2}$ and $n(C_1)+n(C_2)< 4^{k_1}2^{k_2}-1$ having generator matrices (in standard form) $G_1$ and $G_2,$ respectively, then the linear code $C'$ generated by the matrix $G':=[G_1|G_2]$ is also optimal of length $n(C_1)+n(C_2).$
\end{cor}

\begin{proof}
    Straightforward from Lemma \ref{konstruk1} and Theorem \ref{Plotkin}.
\end{proof}

By using a similar idea as in proving Theorem \ref{konstruk}, we obtain the following result.

\begin{thm}[Construction B]\label{konstruk2}
	If $C$ is a linear code of parameters $[n,4^{k_1}2^{k_2},d_L]$ with generator matrix (in standard form) $G,$ then the code $C'_1$ generated by the matrix
	\begin{equation}\label{matrix1}
	G'_1:=\arraycolsep2.5 pt \left[\begin{array}{@{}c|c|c|c@{}}
		G & G & G & G \\
		\hline
		0\;0\dots 0 &1\;1\dots 1 &2\;2\dots 2 &3\;3\dots 3
	\end{array}\;\right]
	\end{equation}
	is linear of parameters   $[4n,4^{k_1+1}2^{k_2},d_L(C'_1)]$ and the code $C'_2$ generated by
	\begin{equation}\label{matrix2}
	    G'_2:=\arraycolsep2.5 pt\left[\begin{array}{@{}c|c@{}}
		G & G \\
		\hline
		0\;0\dots 0  &2\;2\dots 2
	    \end{array}\;\right]
	\end{equation}
	is linear of parameters $[2n,4^{k_1}2^{k_2+1},d_L(C'_2)],$ with $d_L(C'_1)\geq \min\{4n,4d_L\}$ and $d_L(C'_2)\geq \min\{2n,2d_L\}.$
\end{thm}

\begin{cor}[Construction 3 for optimal codes]\label{opt3}
	If $C$ is a Plotkin-optimal linear code of parameters $[n,4^{k_1}2^{k_2},d_L]$ and $n(C)< 4^{k_1}2^{k_2}-1$ with generator matrix (in standard form) $G,$ then the linear code $C'_1$ generated by the matrix $G'_1$ (\ref{matrix1})
	is optimal with parameters $[4n,4^{k_1+1}2^{k_2},4n]$ and the linear code $C'_2$ generated by the matrix $G'_2$ (\ref{matrix2})
	is also optimal with parameters $[2n,4^{k_1}2^{k_2+1},2n].$
\end{cor}

\begin{proof}
	
Straightforward from Theorem \ref{Plotkin} and Theorem \ref{konstruk2}.
\end{proof}

It is easy to see that the linear codes of parameters $[1,4^1 2^0, 1]$ and $[4,4^2 2^0,4]$ are both Plotkin-optimal. By applying Construction 3 as given in Corollary \ref{opt3} repeatedly, we obtain a class of optimal codes as follows.

\begin{cor}\label{opt4}
	Let $k_1,k_2$ be nonnegative integers with $k_1+k_2>0$. Then there exist an optimal linear code of parameters $[4^{k_1}2^{k_2},$ $4^{k_1+1}2^{k_2},4^{k_1}2^{k_2}]$.
\end{cor}

\begin{rem}
	\emph{Zinoviev and Zinoviev \cite{Zinoviev} have also constructed codes over $\ZZ_4$ of parameters $(n,4n,n)$. Their construction utilizes Hadamard matrices of order $n.$ However, our construction above is simpler and produces linear codes.}
\end{rem}


Next, we derive another method to construct optimal codes.
\begin{thm}\label{konstruk4}
	The linear code $C$ having a generator matrix $G$ given by
	\[
	G:=\arraycolsep2.5 pt
    \left[\begin{array}{@{}c|c|c| c c c@{}}
		G^{(k_1,0)} & G^{(k_1,0)} & G^{(k_1,0)} &\mathbf{0} &\mathbf{0} &\mathbf{0} \\\hline
		1\;1\dots 1 &2\;2\dots 2 &3\;3\dots 3 &1 &2 &3
	\end{array}\;\right]
	\]
	is optimal of parameters $[3.4^{k_1},4^{k_1+1}2^0,3.4^{k_1}]$.
\end{thm}

\begin{proof}
	Since the code $C^{(k_1+1,0)}$ has a generator matrix
	\[
	G^{(k_1+1,0)}:=
	\arraycolsep2.5 pt
    \left[\begin{array}{@{}c|c|c|c| c c c@{}}
		G^{(k_1,0)} & G^{(k_1,0)} & G^{(k_1,0)} & G^{(k_1,0)} &\mathbf{0} &\mathbf{0} &\mathbf{0} \\\hline
		0\;0\dots 0 &1\;1\dots 1 &2\;2\dots 2 &3\;3\dots 3 &1 &2 &3
	\end{array}\;\right],\]
	then for any $\mathbf{c}\in C,$ we have $w_L(\mathbf{c})=w_L(\mathbf{c_1})-w_L(\mathbf{c_2})$
	for some $\mathbf{c_1}\in C^{(k_1+1,0)}$ and $\mathbf{c_2}\in C^{(k_1,0)}.$ Since the Lee weight of all nonzero codewords in $C^{(k_1+1,0)}$ is $4^{k_1+1}$ and the Lee weight of all nonzero codewords in $C^{(k_1,0)}$ is $4^{k_1}$ (Theorem \ref{konstruk}), then the Lee weight of any codeword in $C$ is either $0,$ $3.4^{k_1},$ or $4^{k_1+1}.$  Therefore, $d_L(C)=3.4^{k_1}.$
\end{proof}

\begin{rem}
	\emph{The Theorem \ref{konstruk4} above gives us a way to obtain free optimal two-weight codes (codes with exactly two nonzero weights).}
\end{rem}

By applying Corollary \ref{opt3} to the code in Theorem \ref{konstruk4}, or by attaching three codes of parameters
$[4^{k_1}2^{k_2},4^{k_1+1}2^{k_2},4^{k_1}2^{k_2}]$ (see Corollary \ref{opt4}) using Corollary \ref{opt2}, we obtain another class of optimal codes.

\begin{cor}
	Let $k_1,k_2$ be nonnegative integers with $k_1+k_2>0$. Then there exist an optimal linear code of parameters $[3.4^{k_1}2^{k_2},$ $4^{k_1+1}2^{k_2},3.4^{k_1}2^{k_2}].$  
\end{cor}

We end this part by deriving a construction method for free optimal linear codes.

\begin{cor}[Construction 4 for optimal codes]\label{opt6}
	If $C$ is a Plotkin-optimal linear code of type $4^{k_1} 2^0,$ length $n<4^{k_1}-1$ with generator matrix (in standard form) $G,$ then the linear code $C'$ with generator matrix
	\[
	G':=\arraycolsep2.5 pt
    \left[
	\begin{array}{@{}c|c|c|c| c c c@{}}
		G & G^{(k_1,0)} & G^{(k_1,0)} & G^{(k_1,0)} &\mathbf{0} & \mathbf{0} & \mathbf{0} \\\hline
		0\;0\dots 0 &1\;1\dots 1 &2\;2\dots 2 &3\;3\dots 3 &1 &2 &3
	\end{array}\;
	\right]
	\]
	is optimal of parameters $[3.4^{k_1}+n,4^{k_1+1}2^{0},3.4^{k_1}+n].$
\end{cor}

\begin{proof}
	Observe that every nonzero codeword $\mathbf{c'}\in C'$ satisfies
	$w_L(\mathbf{c'})=w_L(\mathbf{c})+w_L(\mathbf{c^*})$
	for some $\cv \in C,$ where $\cv^*$ is a codeword of Lee weight $3.4^{k_1}$ or $4^{k_1+1}$ (Theorem \ref{konstruk4}). Here, if $w_L(\cv)=0$ then $w_L(\cv^*)=4^{k_1+1},$ which implies
	$w_L(\cv')=w_L(\cv)+w_L(\cv^*)=4^{k_1+1}$.
	Moreover, if $w_L(\cv)>0,$ then
	$w_L(\cv')=w_L(\cv)+w_L(\cv^*)\geq n+3.4^{k_1}.$
	Hence, $d_L\geq 3.4^{k_1}+n.$ By Theorem \ref{Plotkin}, $C'$ is optimal.
\end{proof}

\subsection{Optimal codes for $k_1=2, ~k_2=0$.}

In this subsection, we determine all optimal codes with $k_1=2, k_2=0.$ In Table \ref{tab1}, we compare the codes constructed by our methods with the codes having the highest minimum Lee distance in the database \cite{Aydin2}. By applying Lemma \ref{konstruk1} on the codes in Table \ref{tab1}, we obtain all optimal codes with $k_1=2, k_2=0$, as presented in Table \ref{tab2}. We conclude that the difference between the bound (\ref{PLDB}) and the minimum Lee distance of optimal codes we derived is at most $1.$
\begin{table}[b]
    \vspace{-10ex}
    \begin{center}
        \caption{All optimal linear codes for $k_1=2, k_2=0$.}
    \label{tab2}
    \begin{tabular}{|c||c|c|c|c|c|}
	\hline
	n & \textbf{15m} & 15m+1 & 15m+2 & 15m+3 & \textbf{15m+4}\\
	\hline
	PLDB & \textbf{16m} & 16m+1 & 16m+2 & 16m+3 & \textbf{16m+4}\\
	\hline
	OPTIMAL &\textbf{16m} & 16m & 16m+1 &16m+2 &\textbf{16m+4}\\
	\hline\hline
	n & 15m+5 & \textbf{15m+6} & 15m+7 & \textbf{15m+8} & 15m+9\\
	\hline
	PLDB & 16m+5 & \textbf{16m+6} & 16m+7 & \textbf{16m+8} & 16m+9\\
	\hline
	OPTIMAL &16m+4 & \textbf{16m+6} & 16m+6 &\textbf{16m+8} &16m+8\\
	\hline\hline
	n & \textbf{15m+10} & 15m+11 & \textbf{15m+12} & \textbf{15m+13} & \textbf{15m+14}\\
	\hline
	PLDB & \textbf{16m+10} & 16m+11 & \textbf{16m+12} & \textbf{16m+13} & \textbf{16m+14}\\
	\hline
	OPTIMAL &\textbf{16m+10} & 16m+10 & \textbf{16m+12} & \textbf{16m+13}& \textbf{16m+14}\\
	\hline
\end{tabular}
\end{center}
\footnotesize
\begin{center}
\vspace{-1ex}
{PLDB: Plotkin-type Lee distance bound (\ref{PLDB}), \textbf{bolded text:} Plotkin-optimal.\par}    
\end{center}

\end{table}
\newpage
\begin{rem}
\emph{In Table \ref{tab2}, we obtained the OPTIMAL codes after proving the non-existence of Plotkin-optimal codes. To prove the nonexistence of Plotkin-optimal linear codes $C$ of length $n\equiv 1,3,5\; (\mathrm{mod}\;15)$, we consider the refined Plotkin bound for binary codes \cite{Huffman} and apply it for $\Phi(C)$, where $\Phi$ is the Gray map defined in \cite{Hammons1994}. It is not hard to prove the nonexistence of Plotkin-optimal linear codes of length $n\equiv 2,7,9,11\; (\mathrm{mod}\;15)$, although the proof is too cumbersome to display here.}
\end{rem}
\begin{table*}[h]
\vspace{-2ex}
\caption{New Plotkin-Optimal and Good Linear Codes for $k_1=2,k_2=0,2\leq n\leq 61$.}
\vspace{-1.5ex}
\label{tab1}\tabcolsep3.3pt\renewcommand{\arraystretch}{0.92}
\begin{center}
    \begin{tabular}{|c|c|c|c|c|c||c|c|c|c|c|c|}
	\hline
	n & Note &DB & Good & PLDB & LDB & n & Note &DB & Good & PLDB & LDB\\
	\hline
	2 &Known \cite{Aydin2}& 1  & 1      & 2  & 1
	&\textbf{32}&\textbf{L \ref{konstruk1}} &\textbf{32} & \textbf{33} & \textbf{34} & \textbf{61} \\
	\hline
	3 &Known \cite{Aydin2}& 2  & 2      & 3  & 3
	&33 &L \ref{konstruk1} &34 & 34     & 35 & 63 \\
	\hline
	\textbf{\textit{4}} &\textbf{\textit{Known \cite{Aydin2}}}& \textbf{\textit{4}}  & \textbf{\textit{4}}    & \textbf{\textit{4}}  & \textbf{\textit{5}}
	&\textbf{\textit{34}} &\textbf{\textit{C \ref{opt1}}}&\textbf{\textit{34}} &\textbf{\textit{36}} & \textbf{\textit{36}} & \textbf{\textit{65}}\\
	\hline
	5 &Known \cite{Aydin2}& 4  & 4      & 5  & 7
	&35 &L \ref{konstruk1} & 36 & 36     & 37 & 67\\
	\hline
	\textbf{\textit{6}} &\textbf{\textit{Known \cite{Dougherty2007}}} & \textbf{\textit{5}} & \textbf{\textit{6}}  & \textbf{\textit{6}}  & \textbf{\textit{9}}
	& \textbf{\textit{36}} &\textbf{\textit{C \ref{opt1}}}& \textbf{\textit{36}} & \textbf{\textit{38}} & \textbf{\textit{38}} & \textbf{\textit{69}} \\
	\hline
	7 &L \ref{konstruk1} & 6  & 6      & 7  & 11
	&37 &L \ref{konstruk1} & 38 & 38     & 39 & 71\\
	
	\hline
	\textbf{\textit{8}} &\textbf{\textit{C \ref{opt2}}} & \textbf{\textit{8}}  & \textbf{\textit{8}} & \textbf{\textit{8}}  & \textbf{\textit{13}}
	& \textbf{\textit{38}} &\textbf{\textit{C \ref{opt1}}}& \textbf{\textit{38}} & \textbf{\textit{40}} & \textbf{\textit{40}} & \textbf{\textit{73}}\\
	\hline
	9 &L \ref{konstruk1} & 8  & 8      & 9  & 15
	&39 &L \ref{konstruk1} & 40 & 40     & 41 & 75\\
	\hline
	\textbf{\textit{10}} &\textbf{\textit{C \ref{opt2}}} & \textbf{\textit{9}}  & \textbf{\textit{10}} & \textbf{\textit{10}}  & \textbf{\textit{17}}
	&\textbf{\textit{40}} &\textbf{\textit{C \ref{opt1}}}& \textbf{\textit{40}} & \textbf{\textit{42}} & \textbf{\textit{42}} & \textbf{\textit{77}}\\
	\hline
	11 &L \ref{konstruk1} & 10 & 10     & 11 & 19
	&41 &L \ref{konstruk1} &42 &42 &43 &79\\
	\hline
	\textbf{\textit{12}} &\textbf{\textit{C \ref{opt2}}} & \textbf{\textit{12}} & \textbf{\textit{12}} & \textbf{\textit{12}} & \textbf{\textit{21}}
	&\textbf{\textit{42}} &\textbf{\textit{C \ref{opt1}}}& \textbf{\textit{42}} & \textbf{\textit{44}} & \textbf{\textit{44}} & \textbf{\textit{81}}\\
	\hline
	\textit{\textbf{13}} &\textit{\textbf{\#}} & \textit{\textbf{12}} & \textit{\textbf{13}} & \textit{\textbf{13}} & \textit{\textbf{23}}
	&\textbf{\textit{43}} &\textbf{\textit{C \ref{opt1}}}& \textbf{\textit{44}} & \textbf{\textit{45}} & \textbf{\textit{45}} & \textbf{\textit{83}} \\
	\hline
	\textbf{\textit{14}} &\textbf{\textit{C \ref{opt2}}} & \textbf{\textit{13}} & \textbf{\textit{14}} & \textbf{\textit{14}} & \textbf{\textit{25}}
	&\textbf{\textit{44}} &\textbf{\textit{C \ref{opt1}}}& \textbf{\textit{44}} & \textbf{\textit{46}} & \textbf{\textit{46}} & \textbf{\textit{85}} \\
	\hline
	\textbf{\textit{15}} &\textbf{\textit{T \ref{konstruk}}} & \textbf{\textit{14}} &
	\textbf{\textit{16}} & \textbf{\textit{16}} & \textbf{\textit{27}}
	&\textbf{\textit{45}} &\textbf{\textit{C \ref{opt1}}}& \textbf{\textit{46}} & \textbf{\textit{48}} & \textbf{\textit{48}} & \textbf{\textit{87}}\\
	\hline
	16 &\# & 16 & 16     & 17 & 29
	&46 &L \ref{konstruk1}& 46 & 48 & 49 & 89\\
	\hline
	\textbf{17} &\textbf{L \ref{konstruk1}} & \textbf{16} & \textbf{17} & \textbf{18} & \textbf{31}
	&\textbf{47} &\textbf{L \ref{konstruk1}}& \textbf{48} & \textbf{49} & \textbf{50} & \textbf{91} \\
	\hline
	\textbf{18} &\textbf{L \ref{konstruk1}} & \textbf{17} & \textbf{18} &\textbf{19} &\textbf{33}
	&\textbf{48} &\textbf{L \ref{konstruk1}}& \textbf{48} & \textbf{50} & \textbf{51} & \textbf{93}\\
	\hline
	\textbf{\textit{19}} &\textbf{\textit{\textbf{\textit{C \ref{opt1}}}}} &\textbf{\textit{18}} & \textbf{\textit{20}} &\textbf{\textit{20}} & \textbf{\textit{35}}
	&\textbf{\textit{49}} &\textbf{\textit{C \ref{opt1}}}& \textbf{\textit{50}} & \textbf{\textit{52}} & \textbf{\textit{52}} & \textbf{\textit{95}}\\
	\hline
	20 &L \ref{konstruk1} & 20 & 20     & 21 & 37
	&\textbf{50} &\textbf{L \ref{konstruk1}} & \textbf{50} & \textbf{52} & \textbf{53} & \textbf{97}\\
	\hline
	\textbf{\textit{21}} &\textbf{\textit{C \ref{opt1}}}& \textbf{\textit{20}} & \textbf{\textit{22}} & \textbf{\textit{22}} & \textbf{\textit{39}}
	&\textbf{\textit{51}} &\textbf{\textit{C \ref{opt1}}}& \textbf{\textit{52}} & \textbf{\textit{54}} & \textbf{\textit{54}} & \textbf{\textit{99}}\\
	\hline
	\textbf{22} &\textbf{L \ref{konstruk1}} & \textbf{21} & \textbf{22} & \textbf{23} & \textbf{41}
	&\textbf{52} &\textbf{L \ref{konstruk1}} & \textbf{52} & \textbf{54} & \textbf{55} & \textbf{101}\\
	\hline
	\textbf{\textit{23} }&\textbf{\textit{C \ref{opt1}}}& \textbf{\textit{22}} & \textbf{\textit{24}} & \textbf{\textit{24}} & \textbf{\textit{43}}
	&\textbf{\textit{53}} &\textbf{\textit{C \ref{opt1}}} & \textbf{\textit{54}} & \textbf{\textit{56}} & \textbf{\textit{56}} & \textbf{\textit{103}}\\
	\hline
	24 &L \ref{konstruk1} & 24 & 24     & 25 & 45
	&\textbf{54} &\textbf{L \ref{konstruk1}} & \textbf{54}& \textbf{56} & \textbf{57} & \textbf{105}\\
	\hline
	\textbf{\textit{25}} &\textbf{\textit{C \ref{opt1}}}& \textbf{\textit{24}} & \textbf{\textit{26}} & \textbf{\textit{26}} & \textbf{\textit{47}}
	&\textbf{\textit{55}} &\textbf{\textit{C \ref{opt1}}} & \textbf{\textit{56}} & \textbf{\textit{58}} & \textbf{\textit{58}} & \textbf{\textit{107}}\\
	\hline
	\textbf{26} &\textbf{L \ref{konstruk1}} & \textbf{25}
	& \textbf{26} & \textbf{27} & \textbf{49}
	&\textbf{56} &\textbf{L \ref{konstruk1}}& \textbf{56} & \textbf{58} & \textbf{59} & \textbf{109}\\
	\hline
	\textbf{\textit{27}} &\textbf{\textit{C \ref{opt1}}}& \textbf{\textit{26}} & \textbf{\textit{28}} & \textbf{\textit{28}} & \textbf{\textit{51}}
	&\textbf{\textit{57}} &\textbf{\textit{C \ref{opt1}}}& \textbf{\textit{58}} & \textbf{\textit{60}} & \textbf{\textit{60}} & \textbf{\textit{111}}\\
	\hline
	\textbf{\textit{28}} &\textbf{\textit{C \ref{opt1}}}& \textbf{\textit{28}}& \textbf{\textit{29}} & \textbf{\textit{29}} & \textbf{\textit{53}}
	&\textbf{\textit{58}} &\textbf{\textit{C \ref{opt1}}}& \textbf{\textit{60}} & \textbf{\textit{61}} & \textbf{\textit{61}} & \textbf{\textit{113}}\\
	\hline
	\textbf{\textit{29}} &\textbf{\textit{C \ref{opt1}}}& \textbf{\textit{30}} & \textbf{\textit{30}}     & \textbf{\textit{30}} & \textbf{\textit{55}}
	&\textbf{\textit{59}} &\textbf{\textit{C \ref{opt1}}}& \textbf{\textit{60}} & \textbf{\textit{62}} & \textbf{\textit{62}} & \textbf{\textit{115}}\\
	\hline
	\textbf{\textit{30}} &\textbf{\textit{C \ref{opt1}}} &\textbf{\textit{30}} & \textbf{\textit{32}} & \textbf{\textit{32}} & \textbf{\textit{57}}
	&\textbf{\textit{60}} &\textbf{\textit{C \ref{opt1}}}& \textbf{\textit{62}} & \textbf{\textit{64}} & \textbf{\textit{64}} & \textbf{\textit{117}} \\
	\hline
	31 &L \ref{konstruk1} & 32 & 32     & 33 & 59
	&61 &L \ref{konstruk1} & 64 & 64     & 65 & 119\\
	\hline
\end{tabular}
\end{center}
\vspace{-1ex}
\footnotesize
{DB: The highest minimum Lee distance among all existing linear codes of length $n$ in the database \cite{Aydin2}, \\C: Corollary, T: Theorem, L: Lemma, PLDB: Plotkin-type Lee distance bound (\ref{PLDB}), LDB: Singleton-type Lee distance bound (\ref{LDB1}), \textbf{\textit{bold-italic text:}} Plotkin-optimal, \textbf{bolded text}: New good linear code,\\ \#: Constructed by adding or removing column(s) from the
generator matrix of the nearest optimal code.\par}
\end{table*}
\vspace{-0.2ex}
\subsection{Optimal codes for $k_1=3, ~k_2=0.$}
In this subsection, we use our methods to construct many new optimal and nearly optimal free linear codes which are unknown to exist before in the database of $\ZZ_4$ codes \cite{Aydin2}, as presented in Table \ref{tab3}. In many cases, the minimum Lee distance improved significantly. From Table \ref{tab3} and Lemma \ref{konstruk1}, we conclude that the difference between the bound (\ref{PLDB}) and the minimum Lee distance of optimal codes is at most 2.

\begin{table*}[h]
    \vspace{-2.5ex}
    \caption{New Plotkin-Optimal and Good Linear Codes for $k_1=3,k_2=0,3\leq n\leq 66$.}
    \vspace{-1.5ex}
    \label{tab3}
    \begin{center}
    \tabcolsep3.3pt\renewcommand{\arraystretch}{0.92}
    \begin{tabular}{|c|c|c|c|c|c||c|c|c|c|c|c|}
	\hline
	n & Note &DB & Good & PLDB & LDB & n & Note &DB & Good & PLDB & LDB \\
	\hline
	3 &Known \cite{Aydin2} & 1 & 1 & 3 & 1
	&\textbf{35} &\textbf{L \ref{konstruk1}} & \textbf{28} & \textbf{34} & \textbf{35} & \textbf{65}\\ \hline
	4 &Known \cite{Dougherty2007} & 1 & 2 & 4 & 3
	&\textbf{36} &\textbf{L \ref{konstruk1}} & \textbf{30} &\textbf{34}  & \textbf{36} & \textbf{67}\\ \hline
	5 &Known \cite{Aydin2} & 3 & 3 & 5 & 5
	&\textbf{37} &\textbf{L \ref{konstruk1}} & \textbf{30} &\textbf{36}  & \textbf{37} & \textbf{69}\\ \hline
	6 &Known \cite{Aydin2} & 4 &4  & 6 & 7
	&\textbf{38} &\textbf{L \ref{konstruk1}} & \textbf{31} & \textbf{36} & \textbf{38} & \textbf{71}\\ \hline
	7 &Known \cite{Aydin2} & 6 & 6 & 7 & 9
	&\textbf{39} &\textbf{L \ref{konstruk1}} & \textbf{32} & \textbf{38} & \textbf{39} & \textbf{73}\\ \hline
	8 &Known \cite{Aydin2} & 6 & 6 & 8 & 11
	&\textbf{\textit{40}} &\textbf{\textit{C \ref{opt2}}} & \textbf{\textit{32}} & \textbf{\textit{40}} & \textbf{\textit{40}} & \textbf{\textit{75}}\\ \hline
	\textbf{9} &\textbf{\#} & \textbf{6} & \textbf{7} & \textbf{9} & \textbf{13}
	&\textbf{41} &\textbf{L \ref{konstruk1}} & \textbf{34} & \textbf{40} & \textbf{41} & \textbf{77}\\ \hline
	\textbf{10} &\textbf{\#} & \textbf{6} &\textbf{8}  & \textbf{10} & \textbf{15}
	&\textbf{42} &\textbf{L \ref{konstruk1}} & \textbf{34} & \textbf{40} & \textbf{42} & \textbf{79}\\ \hline
	\textbf{11} & \textbf{\#} & \textbf{8} & \textbf{9} & \textbf{11} & \textbf{17}
	&\textbf{43} &\textbf{L \ref{konstruk1}} & \textbf{34} &\textbf{42}  & \textbf{43} & \textbf{81}\\ \hline
	12 &Known \cite{Aydin2} & 10 & 10 & 12 & 19
	&\textbf{\textit{44}} &\textbf{\textit{C \ref{opt2}}} & \textbf{\textit{36}} & \textbf{\textit{44}} & \textbf{\textit{44}} & \textbf{\textit{83}} \\ \hline
	13 &Known \cite{Aydin2} & 11 & 11 & 13 & 21
	&\textbf{45} &\textbf{L \ref{konstruk1}} & \textbf{38} & \textbf{44} & \textbf{45} & \textbf{85}\\ \hline
	14 &L \ref{konstruk1} & 12 & 12 & 14 & 23
	&\textbf{46} &\textbf{L \ref{konstruk1}} & \textbf{38} &\textbf{44}  & \textbf{46} & \textbf{87}\\ \hline
	\textbf{15} &\textbf{\#} & \textbf{10} & \textbf{14}  & \textbf{15} & \textbf{25}
	&\textbf{47} &\textbf{L \ref{konstruk1}} & \textbf{39} &\textbf{46}  & \textbf{47} & \textbf{89} \\ \hline
	\textbf{\textit{16}} &\textbf{\textit{C \ref{opt3}}} & \textbf{\textit{12}} & \textbf{\textit{16}} & \textbf{\textit{16}} & \textbf{\textit{27}}
	&\textbf{\textit{48}} &\textbf{\textit{C \ref{opt2}}} & \textbf{\textit{40}} & \textbf{\textit{48}} & \textbf{\textit{48}} & \textbf{\textit{91}} \\ \hline
	\textbf{17} &\textbf{\#} & \textbf{12} & \textbf{16} & \textbf{17} & \textbf{29}
	&\textbf{49} &\textbf{L \ref{konstruk1}} & \textbf{42} & \textbf{48} & \textbf{49} & \textbf{93}\\ \hline
	\textbf{18} &\textbf{\#} & \textbf{12} & \textbf{16} & \textbf{18} & \textbf{31}
	&\textbf{50} &\textbf{\#} & \textbf{44} & \textbf{49} & \textbf{50} & \textbf{95} \\ \hline
	\textbf{19} &\textbf{\#} & \textbf{14} & \textbf{18}  & \textbf{19} & \textbf{33}
	&\textbf{51} &\textbf{L \ref{konstruk1}} & \textbf{44} &\textbf{50}  & \textbf{51} & \textbf{97}\\
	\hline
	\textbf{20} &\textbf{L \ref{konstruk1}} & \textbf{16} & \textbf{18} & \textbf{20} & \textbf{35}
	&\textbf{\textit{52}}&\textbf{\textit{C \ref{opt6}}} & \textbf{\textit{44}} & \textbf{\textit{52}} & \textbf{\textit{52}} & \textbf{\textit{99}}\\
	\hline
	\textbf{21} &\textbf{\#} & \textbf{18} & \textbf{20} & \textbf{21} & \textbf{37}
	&\textbf{53} &\textbf{L \ref{konstruk1}} & \textbf{45} & \textbf{52} & \textbf{53} & \textbf{101}\\
	\hline
	\textbf{22} &\textbf{L \ref{konstruk1}} & \textbf{17} & \textbf{20} & \textbf{22} & \textbf{39}
	&\textbf{\textit{54}} &\textbf{\textit{C \ref{opt6}}}& \textbf{\textit{45}} & \textbf{\textit{54}} & \textbf{\textit{54}} & \textbf{\textit{103}}\\
	\hline
	\textbf{23} &\textbf{L \ref{konstruk1}} & \textbf{18} & \textbf{22} & \textbf{23} & \textbf{41}
	&\textbf{55} &\textbf{L \ref{konstruk1}} & \textbf{47} &\textbf{54}  & \textbf{55} & \textbf{105}\\
	\hline
	\textbf{\textit{24}} &\textbf{\textit{C \ref{opt3}}} & \textbf{\textit{20}} & \textbf{\textit{24}} & \textbf{\textit{24}} & \textbf{\textit{43}}
	&\textbf{\textit{56}} &\textbf{\textit{C \ref{opt6}}}& \textbf{\textit{47}} & \textbf{\textit{56}} & \textbf{\textit{56}} & \textbf{\textit{107}}\\
	\hline
	\textbf{25} &\textbf{\#} & \textbf{22} &\textbf{24}  & \textbf{25} & \textbf{45}
	&\textbf{57} &\textbf{L \ref{konstruk1}} & \textbf{48} &\textbf{56}  & \textbf{57} & \textbf{109}\\
	\hline
	\textbf{26} &\textbf{L \ref{konstruk1}} & \textbf{22} & \textbf{24} & \textbf{26} & \textbf{47}
	&\textbf{\textit{58}} &\textbf{\textit{C \ref{opt6}}}& \textbf{\textit{50}} & \textbf{\textit{58}} & \textbf{\textit{58}} & \textbf{\textit{111}}\\
	\hline
	\textbf{27} &\textbf{\#} & \textbf{22} &\textbf{26} & \textbf{27} & \textbf{49}
	&\textbf{59} &\textbf{L \ref{konstruk1}} & \textbf{50} &\textbf{58}  & \textbf{59} & \textbf{113}\\
	\hline
	\textbf{\textit{28}} &\textbf{\textit{\#}} & \textbf{\textit{23}} & \textbf{\textit{28}} & \textbf{\textit{28}} & \textbf{\textit{51}}
	&\textbf{\textit{60}} &\textbf{\textit{C \ref{opt6}}}& \textbf{\textit{50}} & \textbf{\textit{60}} & \textbf{\textit{60}} & \textbf{\textit{115}}\\
	\hline
	\textbf{29} &\textbf{\#} & \textbf{23} & \textbf{28} & \textbf{29} & \textbf{53}
	&\textbf{\textit{61}} &\textbf{\textit{C \ref{opt6}}}& \textbf{\textit{52}} & \textbf{\textit{61}} & \textbf{\textit{61}} & \textbf{\textit{117}}\\
	\hline
	\textbf{30} &\textbf{L \ref{konstruk1}} & \textbf{25} & \textbf{28} & \textbf{30} & \textbf{55}
	&\textbf{\textit{62}} &\textbf{\textit{C \ref{opt6}}}& \textbf{\textit{52}} & \textbf{\textit{62}} & \textbf{\textit{62}} & \textbf{\textit{119}}\\
	\hline
	\textbf{31} &\textbf{L \ref{konstruk1}} & \textbf{25} & \textbf{30} & \textbf{31} & \textbf{57}
	&\textbf{\textit{63}} &\textbf{\textit{T \ref{konstruk}}}& \textbf{\textit{53}} & \textbf{\textit{64}} & \textbf{\textit{64}} & \textbf{\textit{121}}\\
	\hline
	\textbf{\textit{32}} &\textbf{\textit{C \ref{opt2}}} & \textbf{\textit{26}} & \textbf{\textit{32}} & \textbf{\textit{32}} & \textbf{\textit{59}}
	&\textbf{64} &\textbf{L \ref{konstruk1}} & \textbf{54} & \textbf{64} & \textbf{65} & \textbf{123}\\
	\hline
	\textbf{33} &\textbf{L \ref{konstruk1}} & \textbf{28} & \textbf{32} & \textbf{33} & \textbf{61}
	&\textbf{65} &\textbf{L \ref{konstruk1}} & \textbf{42} & \textbf{64} & \textbf{66} & \textbf{125} \\
	\hline
	\textbf{34} &\textbf{L \ref{konstruk1}} & \textbf{28} & \textbf{32} & \textbf{34} & \textbf{63}
	&\textbf{66} &\textbf{L \ref{konstruk1}} & \textbf{44} &\textbf{65}  & \textbf{67} & \textbf{127}\\
	\hline
	\end{tabular}    
    \end{center}
    
	\vspace{-1.0ex}
	\footnotesize
	{DB: The highest minimum Lee distance among all existing linear codes of length $n$ in the database \cite{Aydin2}, \\C: Corollary, T: Theorem, L: Lemma, PLDB: Plotkin-type Lee distance bound (\ref{PLDB}), LDB: Singleton-type Lee distance bound (\ref{LDB1}), \textbf{\textit{bold-italic text:}} Plotkin-optimal, \textbf{bolded text}: New good linear code,\\ \#: Constructed by adding or removing column(s) from the
generator matrix of the nearest optimal code.\par}
\end{table*}
\newpage
\section{Concluding remarks}
In this work, we introduced several construction methods for linear codes over $\ZZ_4.$  By using these methods, we obtained many new optimal linear codes as well as new linear codes with the highest known minimum distances compared with the known codes available in the database of $\ZZ_4$ codes \cite{Aydin2009, Aydin2}, as presented in Table \ref{tab1} and Table \ref{tab3}. We can apply the Gray map defined in \cite{Hammons1994} to our new Plotkin-optimal linear codes to obtain many binary codes that are also Plotkin-optimal and $\ZZ_4$-linear, which is interesting to investigate further.

Regarding the existence of linear codes with good minimum Lee distance, we believe in the following conjecture.

\begin{conj}
	There exist free linear codes $C$ over $\mathbb{Z}_4$ of parameters $[n,4^{k}2^0,d_L],$ with the minimum distance satisfying
	$\displaystyle d_L\geq \left \lfloor\frac{4^k}{4^k-1}n \right\rfloor-(ak+b)$
	for some constant $a, b$.
\end{conj}

\noindent For $a=1, b=0$, we notice that the conjecture holds for all $n$ when $k=2,3$. For $k>3$, it holds for $n\equiv m \pmod {(4^k-1)},$ with $0\leq m\leq k+1$ and $3\cdot 4^{k-1}\leq m< 4^k-1$.

We are now working on a generalization of our observation here to the linear codes over the ring $\ZZ_{2^r}.$ So far, we succeed to prove that several properties here also hold for codes over the ring $\ZZ_{2^r},$ with $r \in \ZZ^+.$ As an example, we can obtain a similar bound as in Theorem \ref{Plotkin} for linear codes over $\ZZ_{2^r}$. This results, which is now in preparation, will be published elsewhere in a separate paper.



\begin{thebibliography}{99}

\bibitem{Abualrub2004b}
T. Abualrub, A. Ghrayeb, and R. H. Oehmke, "A mass formula and rank of $\ZZ_4$ cyclic codes of length $2^e,$"
\emph{IEEE Trans. Inform. Theory} {\bf 50}(12) (2004), 3306-3312.

\bibitem{Abualrub2003}
T. Abualrub and R. H. Oehmke, "Cyclic codes of length $2^e$ over $\ZZ_4$,"
\emph{Discr. Appl. Math.} {\bf 128}(1 SPEC.) (2003), 3-9.



\bibitem{Aydin2009}
N. Aydin and T. Asamov, "A database of $\ZZ_4$ codes," \emph{J. of Combinatorics, Information \& System Sciences} {\bf 34}(1-4) (2009), 1-12.

\bibitem{Aydin2}
N. Aydin and T. Asamov, "Online database of $\ZZ_4$ codes," available at
http://quantumcodes.info/Z4 (accessed at February 5, 2022).

\bibitem{Blackford2003a}
T. Blackford, "Cyclic codes over $\ZZ_4$ of oddly even length,"
\emph{Discr. Appl. Math.} {\bf 128}(1 SPEC.) (2003), 27-46.

\bibitem{Blackford2003b}
T. Blackford, "Negacyclic codes over $\ZZ_4$ of even length,"
\emph{IEEE Trans. Inform. Theory} {\bf 49}(6) (2003), 1417-1424.

\bibitem{Blake1}
I. F. Blake, "Codes over certain rings," \emph{Inf. Control} {\bf 20} (1972), 396-404.

\bibitem{Blake2}
I. F. Blake, "Codes over integer residue rings," \emph{Inf. Control} {\bf 29} (1975), 295-300.





\bibitem{Bonnecaze1}
A. Bonnecaze and I. Duursma, "Translate of linear codes over $\ZZ_4,$"
\emph{IEEE Trans. Inform. Theory} {\bf 43}(4) (1997), 1218-1230.

\bibitem{Bustomi21}
Bustomi, A. P. Santika, and D. Suprijanto, "Linear codes over the ring $\ZZ_4+u\ZZ_4+v\ZZ_4+w\ZZ_4+uv\ZZ_4+uw\ZZ_4+vw\ZZ_4+uvw\ZZ_4,$" 
\emph{IAENG Int. J. Comput. Science} {\bf 43}(8) (2021), 11 pages.

\bibitem{Byrne2013}
E. Byrne, M. Greferath, J. Pernas, and J. Zumbr\"{a}gel,
"Algebraic decoding of negacyclic codes over $\ZZ_4,$"
\emph{Des. Codes Cryptogr.} {\bf 66}(1-3) (2013), 3-16.

\bibitem{Dougherty2007}
S. T. Dougherty, T. A. Gulliver, Y. H. Park, and J. N. C. Wong, "Optimal linear codes over $\ZZ_m,$"
\emph{J. Korean Math. Soc.} {\bf 44}(5) (2007), 1139-1162.

\bibitem{Dougherty2006}
S. T. Dougherty and S. Ling, "Cyclic codes over $\ZZ_4$ of even length,"
\emph{Des. Codes Cryptogr.} {\bf 39}(2) (2006), 127-153.

\bibitem{Dougherty2001}
S. T. Dougherty and K. Shiromoto, "Maximum distance codes over rings of order $4,$"
\emph{IEEE Trans. Inform. Theory} {\bf 47}(1) (2001), 400-404.



\bibitem{Gulliver2007}
T. A. Gulliver and J. N. C. Wong, "Classification of optimal linear $\ZZ_4$ rate $1/2$ codes of length $\geq 8,$"
\emph{Ars Combin.} {\bf 85} (2007), 287-306

\bibitem{Hammons1994}
A. R. Hammons, P. V. Kumar, A. R. Calderbank, N. J. A. Sloane, and P. Sol\'{e},
"The $\ZZ_4$-linearity of Kerdock, Preparata, Goethals and related codes,"
\emph{IEEE Trans. Inform. Theory} {\bf 40}(2) (1994), 301-319.


\bibitem{Huffman2005}
W. C. Huffman, "On the classification and enumeration of self-dual codes,"
\emph{Finite Fields Appl.} {\bf 11}(3) (2005), 451-490.

\bibitem{Huffman}
W. C. Huffman and V. Pless, \emph{Fundamentals of Error Correcting Codes,} Cambridge University Press, 2003.


\bibitem{Klove}
T. Kl{\o}ve, "Support weight distribution of linear codes," \emph{Discrete Math.} {\bf 106-107} (1992), 311-316.


\bibitem{Pless2} V. S. Pless and Z. Qiang, "Cyclic codes and quadratic residue codes over $\ZZ_4,$" \emph{IEEE Trans. Inform. Theory}
{\bf 42}(5) (1996), 1594-1600.

\bibitem{Rains1999}
E. M. Rains, "Optimal self-dual codes over $\ZZ_4,$"
\emph{Discrete Math.} {\bf 203}(1-3) (1999), 215-228.

\bibitem{Spiegel1}
E. Spiegel, "Codes over $\ZZ_m$," \emph{Inf. Control} {\bf 35} (1977), 48-51.

\bibitem{Spiegel2}
E. Spiegel, "Codes over $\ZZ_m,$ revisited," \emph{Inf. Control} {\bf 37} (1978), 100-104.

\bibitem{Vega2004}
G. Vega and J. Wolfmann, "Some families of $\ZZ_4$-cyclic codes,"
\emph{Finite Fields Appl.} {\bf 10}(4) (2004), 530-539.

\bibitem{Wan}
Z. Wan, \emph{Quaternary codes,} World Scientific, 1997.



\bibitem{Wyner1968}
A. D. Wyner and R. L. Graham, "An upper bound on minimum distance for a $k$-ary code,"
\emph{Inf. Control} {\bf 13}(1) (1968), 46-52.


\bibitem{Zinoviev}
V. A. Zinoviev and D. V. Zinoviev, "On the generalized concatenated construction
for codes in $L_1$ and Lee metrics," \emph{Probl. Inf. Transm.} {\bf 57}(1) (2021), 70-83.


\end{thebibliography}

\end{document}